\documentclass[runningheads]{llncs}
\usepackage[T1]{fontenc}
\usepackage{graphicx}
\usepackage{hyperref}
\usepackage{amsfonts}
\usepackage{amsmath}
\usepackage{amssymb}
\usepackage{bussproofs}
\usepackage{caption}
\usepackage{comment}
\usepackage{enumitem}
\usepackage{fancyvrb}
\usepackage{framed}
\usepackage{mathtools}
\usepackage{stmaryrd}
\usepackage{tikz-cd}
\usepackage{bbding}

\newcommand{\fusion}{\diamond}
\newcommand{\gkat}{\mathtt{GKAT}}
\newcommand{\kat}{\mathtt{KAT}}
\newcommand{\wh}{\mathsf{wh}}
\newcommand{\wwh}{\mathsf{wwh}}
\newcommand{\lb}{\scalebox{0.8}{$\square$}}
\renewcommand{\int}[1]{\llbracket{#1}\rrbracket}

\begin{document}
\title{A cyclic proof system for Guarded Kleene Algebra with Tests (full version)\thanks{The research of Jan Rooduijn has been made possible by a grant from the Dutch Research Council NWO, project number 617.001.857.} }
%
%
\author{Jan Rooduijn\inst{1}\Envelope \and
Dexter Kozen\inst{2} \and
Alexandra Silva\inst{2}}
\authorrunning{J. Rooduijn et al.}
%
\institute{Institute of Logic, Language and Computation, University of Amsterdam, The Netherlands
\email{janrooduijn@gmail.com}\\ \and
Cornell University, Ithaca, NY, USA 
}
\maketitle              

\begin{abstract}
Guarded Kleene Algebra with Tests ($\gkat$ for short) is an efficient fragment of Kleene Algebra with Tests, suitable for reasoning about simple imperative while-programs. Following earlier work by Das and Pous on Kleene Algebra, we study $\gkat$ from a proof-theoretical perspective. The deterministic nature of $\gkat$ allows for a non-well-founded sequent system whose set of regular proofs is complete with respect to the guarded language model. This is unlike the situation with Kleene Algebra, where hypersequents are required. Moreover, the decision procedure induced by proof search runs in  $\mathsf{NLOGSPACE}$, whereas that of Kleene Algebra is in $\mathsf{PSPACE}$.

\end{abstract}
\section{Introduction}
    Guarded Kleene Algebra with Test ($\gkat$) is the fragment of Kleene Algebra with Tests ($\kat$) comprised of the deterministic $\mathsf{while}$ programs. Those are the programs built up from sequential composition $(e \cdot f)$, conditional branching (\texttt{if-$b$-then-$e$-else-$f$}) and loops (\texttt{while $b$ do $e$}). For an introduction to $\kat$ we refer the reader to~\cite{KozenKAT97}. The first papers focusing on the fragment of $\kat$ that is nowadays called $\gkat$ are Kozen's~\cite{Kozen08} and Kozen \& Tseng's~\cite{KozenT08}, where it is used to study the relative power of several programming constructs. 

    As $\gkat$ is a fragment of $\kat$, it directly inherits a rich theory. It admits a language semantics in the form of \emph{guarded strings} and for every expression there is a corresponding $\kat$-automaton. Already in~\cite{KozenT08} it was argued that $\gkat$ expressions are more closely related to so-called \emph{strictly deterministic automata}, where every state transition executes a primitive program. Smolka et al. significantly advanced the theory of $\gkat$ in~\cite{Smolka20}, by studying various additional semantics, identifying the precise class of strictly deterministic automata corresponding to $\gkat$-expressions (proving a \emph{Kleene theorem}), giving a nearly linear decision procedure of the equivalence of $\gkat$-expressions, and studying its equational axiomatisation. Since then $\gkat$ has received considerable further attention, \emph{e.g.} in~\cite{SchmidKK021,RozowskiKKS023,Zetzsche22,schmid2023complete}.

    One of the most challenging and intriguing aspects of $\gkat$ is its proof theory. The standard equational axiomatisation of $\kat$ from~\cite{KozenKAT97} does not simply restrict to $\gkat$, since a derivation of an expression that lies within the $\gkat$-fragment might very well contain expressions that lie outside of it.  Moreover, the axiomatisation of $\kat$ contains a least fixed point rule that relies on the equational definability of inclusion, which does not seem to be available in $\gkat$.
    
    In~\cite{Smolka20}, this problem is circumvented by introducing a custom equational axiomatisation for $\gkat$ that uses a \emph{unique} fixed point rule. While a notable result, this solution is still not entirely satisfactory. First, completeness is only proven under the inclusion of a variant of the unique fixed point rule that operates on entire systems of equations (this problem was recently addressed for the so-called \emph{skip-free} fragment of $\gkat$ in~\cite{schmid2023complete}). Moreover, even the ordinary, single-equation, unique fixed point rule contains a non-algebraic side-condition, analogous to the empty word property in Salomaa's axiomatisation of Kleene Algebra~\cite{salomaa1966two}. Because of this, a proper definition of `a $\gkat$' is still lacking. 

    In recent years the proof theory of logic with fixed point operators (such as \texttt{while-$b$-do-$e$}) has seen increasing interest in \emph{non-well-founded} proofs. In such proofs, branches need not to be closed by axioms, but may alternatively be infinitely deep. To preserve soundness, a progress condition is often imposed on each infinite branch, facilitating a soundness proof by infinite descent. In some cases non-well-founded proofs can be represented by finite trees with back-edges, which are then called \emph{cyclic proofs}. See \emph{e.g.} ~\cite{brotherston2005cyclic,kuperberg2021cyclic,Wehr22,afshari2022cyclicfirstorder,dekker2023proof} for a variety of such approaches. Often, the non-well-founded proof theory of some logic is closely related to its corresponding automata theory. Taking the proof-theoretical perspective, however, can be advantageous because it is more fine-grained and provides a natural setting for establishing results such as interpolation~\cite{MartiVenema21arxiv,Gui1}, cut elimination~\cite{SavateevS17,AcclavioCG24}, and completeness by proof transformation~\cite{SprengerD03,Das18}.

    In~\cite{DasPous17}, Das \& Pous study the non-well-founded proof theory of Kleene Algebra, a close relative of $\gkat$ (for background on Kleene Algebra we refer the reader to~\cite{kozen1994completeness}). They show that a natural non-well-founded sequent system for Kleene Algebra is not complete when restricting to the subset of cyclic proofs. To remedy this, they introduce a \emph{hypersequent} calculus, whose cyclic proofs \emph{are} complete. They give a proof-search procedure for this calculus and show that it runs in $\mathsf{PSPACE}$. Since deciding Kleene Algebra expressions is $\mathsf{PSPACE}$-complete, their proof-search procedure induces an optimal decision procedure for this problem. In a follow-up paper together with Doumane, left-handed completeness of Kleene Algebra is proven by translating cyclic proofs in the hypersequent calculus to well-founded proofs in left-handed Kleene Algebra~\cite{Das18}. 

    The goal of the present paper is to study the non-well-founded proof theory of $\gkat$. This is interesting in its own right, for instance because, as we will see, it has some striking differences with Kleene Algebra. Moreover, we hope it opens up new avenues for exploring the completeness of algebraic proof systems for $\gkat$, through the translation of our cyclic proofs.
    \newpage
    \subsubsection{Outline} Our paper is structured as follows.
    \begin{itemize}
		\item In Section \ref{sec:prelim} we introduce preliminary material: the syntax of $\gkat$ and its language semantics.
		\item Section \ref{sec:thesgkatsystem} introduces our non-well-founded proof system $\mathsf{SGKAT}$ for $\gkat$. 
		\item In Section \ref{sec:sgkatsoundess} we show that (possibly infinitary) proofs in $\mathsf{SGKAT}$ are sound. That is, the interpretation of each derivable sequent - a $\gkat$-\emph{inequality} - is true in the language model (which means that a certain \emph{inclusion} of languages holds). 
		\item In Section \ref{sec:finite-stateness} we show that proofs are \emph{finite-state}: each proof contains only finitely many distinct sequents. More precisely, by employing a more fine-grained analysis than in~\cite{DasPous17}, we give a quadratic bound on the number of distinct sequents occurring in a proof, in terms of the size of its endsequent. It follows that the subset of cyclic proofs proves exactly the same sequents as the set of all non-well-founded proofs.
        \item Section 6 deals with completeness and complexity. We first use a proof-search procedure to show that $\mathsf{SGKAT}$ is complete: every sequent whose interpretation is valid in the language model, can be derived. We then show that this proof-search procedure runs in $\mathsf{coNLOGSPACE}$. This gives an $\mathsf{NLOGSPACE}$ upper bound on the complexity of the language inclusion problem for $\gkat$-expressions.
	\end{itemize}

    \subsubsection{Our contributions} Our paper closely follows the treatment of Kleene Algebra in~\cite{DasPous17}. Nevertheless, we make the following original contributions:
    \begin{itemize}
    \item Structure of sequents: we devise a form of sequents bespoke to $\gkat$, by labelling the sequents by sets of atoms. This is similar to how the appropriate automata for $\gkat$ are not simply $\kat$-automata. In contrast to Kleene Algebra, it turns out that we do not need to extend our sequents to hypersequents in order to obtain completeness for the fragment of cyclic proofs.

    \item Soundness argument: our modest contribution here is the notion of priority of rules and the fact that our rules are all invertible when they have priority. The soundness argument for finite proofs is, of course, slightly different, because our rules are different. (The step from the soundness of finite proofs, towards the soundness of infinite proofs, is completely analogous to that of~\cite{DasPous17}.)

    \item Regularity: this concerns showing that every proof contains only finitely many distinct sequents. As in~\cite{DasPous17}, our argument views each expression in a proof as a subexpression of an expression in the proof’s root. A modest contribution is that our argument is made more formal by considering these expressions as nodes in a syntax tree. More importantly, the bound on the number of distinct cedents we obtain is sharper: where in~\cite{DasPous17} it is exponential in the size of the syntax tree, our bound is linear (yielding a quadratic bound on the number of sequents). 
    
    \item Completeness: the structure of the argument is identical to that in~\cite{DasPous17}, but the details differ due to the different rules and different type of sequents. This for instance shows in our proof of Lemma \ref{lem:right-productivity} (which is analogous to Lemma 20 in~\cite{DasPous17}), where we make crucial use of the set of atoms annotating a sequent.
    
    \item Complexity: our complexity argument is necessarily different because it applies to a different system and is designed to give a different upper bound.
    \end{itemize}
	Due to space limitations several proofs are only sketched or omitted entirely. Full versions of these proofs can be found in the extended version of this paper~\cite{full}. 
	\section{Preliminaries}
	\label{sec:prelim}
	\subsection{Syntax}
	The language of $\gkat$ has two sorts, namely \emph{programs}\index{program} and a subset thereof consisting of \emph{tests}\index{test}. It is built from a finite and non-empty set $T$ of \emph{primitive tests}\index{test!primitive} and a non-empty set $\Sigma$ of \emph{primitive programs}\index{program!primitive}, where $T$ and $\Sigma$ are disjoint. For the rest of this paper we fix such sets $T$ and $\Sigma$. We reserve the letters $t$ and $p$ to refer, respectively, to arbitrary primitive tests and primitive programs. The first of the following grammars defines the \emph{tests}, and the second the \emph{expressions}.
	\begin{align*}
	b,c ::= 0 \mid 1 \mid t \mid \overline b \mid b \lor c \mid b \cdot c&&
	e,f ::= b \mid p \mid e \cdot f \mid e +_b f \mid e^{(b)},
	\end{align*}
	where $t \in T$ and $p \in \Sigma$. Intuitively, the operator $+_b$ stands for the \texttt{if-then-else} construct, and the operator $(-)^{(b)}$ stands for the \texttt{while} loop. Note that  the tests are simply propositional formulas. It is convention to use $\cdot$ instead of $\land$ for conjunction. As usual, we often omit $\cdot$ for syntactical convenience, \emph{e.g.} by writing $pq$ instead of $p \cdot q$. 
	\begin{example}
	\label{exmp:program}
	The idea of $\gkat$ is to model imperative programs. For instance, the expression $(p +_b q)^{(a)}$ represents the following imperative program: 
	\begin{center}
		\begin{BVerbatim}
		while a do (if b then p else q)
		\end{BVerbatim} 
		\phantom{aaaaaaaaaaa}
	\end{center}
	\end{example}
	\begin{remark}
	\label{rem:katgkatsyn}
	As mentioned in the introduction, $\gkat$ is a fragment of Kleene Algebra with Tests, or $\kat$~\cite{KozenKAT97}\index{$\kat$}. The syntax of $\kat$ is the same as that of $\gkat$, but with unrestricted union $+$ instead of guarded union $+_b$, and unrestricted iteration $(-)^*$ instead of the while loop operator $(-)^{(b)}$. The embedding $\varphi$ of $\gkat$ into $\kat$ acts on guarded union and guarded iteration as follows, and commutes with all other operators: $\varphi(e +_b f) = b \cdot \varphi(e) + \overline b \cdot \varphi(f)$, and  $\varphi(e^{(b)}) = (b \cdot \varphi(e))^* \cdot \overline b$.
	\end{remark}
	\subsection{Semantics}
	There are several kinds of semantics for $\gkat$. In~\cite{Smolka20}, a \emph{language} semantics, a \emph{relational} semantics, and a \emph{probabilistic} semantics are given. In this paper we are only concerned with the language semantics, which we shall now describe.  

	We denote by $\mathsf{At}$\index{$\mathsf{At}$} the set of \emph{atoms}\index{atom} of the free Boolean algebra generated by $T = \{t_1, \ldots t_n\}$. That is, $\mathsf{At}$ consists of all tests of the form $c_1 \cdots c_n$, where $c_i \in \{t_i, \overline{t_i}\}$ for each $1 \leq i \leq n$. Lowercase Greek letters $(\alpha, \beta, \gamma, \ldots)$ will be used to denote elements of $\mathsf{At}$. A \emph{guarded string}\index{guarded string} is an element of the regular set $\mathsf{At} \cdot (\Sigma \cdot \mathsf{At})^*$. That is, a string of the form $\alpha_1 p_1 \alpha_2 p_2 \cdots \alpha_n p_n \alpha_{n+1}$. We will interpret expressions as languages (formally just sets) of guarded strings. The sequential composition operator $\cdot$ is interpreted by the \emph{fusion product} $\diamond$, given by $L \diamond K := \{x \alpha y \mid x\alpha \in L \text{ and } \alpha y \in K\}$. For the interpretation of $+_b$, we define for every set of atoms $B \subseteq \mathsf{At}$ the following operation of \emph{guarded union} on languages: $L +_B K := (B \diamond L) \cup (\overline B \diamond K)$, where $\overline B$ is $\mathsf{At} \setminus B$. For the interpretation of $(-)^{(b)}$, we stipulate:
	\begin{align*}
	L^0 := \mathsf{At} && L^{n+1} := L^{n} \diamond L && L^B := \bigcup_{n \geq 0} (B \fusion L)^n \diamond \overline B 
	\end{align*}
	Finally, the semantics of $\gkat$ is inductively defined as follows:
	\begin{align*}
	&\int{b} := \{\alpha \in \mathsf{At} : \alpha \leq b\} &&\int{p} := \{\alpha p \beta: \alpha, \beta \in \mathsf{At}\}  &&\int{e \cdot f} := \int{e} \diamond \int{f}  \\
	& \int{e +_b f} := \int{e} +_{\int{b}} \int{f}  && \int{e^{(b)}} := \int{e}^{\int{b}}
	\end{align*}
	Note that the interpretation of $\cdot$ between tests is the same whether they are regarded as tests or as programs, \emph{i.e.} $\int{b} \cap \int{c} = \int{b} \diamond \int{c}$.
	\begin{remark}
		While the semantics of expressions is explicitly defined, the semantics of tests is derived implicitly through the free Boolean algebra generated by $T$. It is standard in the $\gkat$ literature to address the Boolean content in this manner.
	\end{remark}
	\begin{example}
	In a guarded string, atoms can be thought of as states of a machine, and programs as executions. For instance, in case of the guarded string $\alpha p \beta$, the machine starts in state $\alpha$, then executes program $p$, and ends in state $\beta$. Let us briefly check which guarded strings of, say, the form $\alpha p \beta q \gamma$ belong to the interpretation $\int{(p +_b q)^{(a)}}$ of the program of Example \ref{exmp:program}. First, we must have $\alpha \leq a$, for otherwise we would not enter the loop. Moreover, we have $\alpha \leq b$, for otherwise $q$ rather than $p$ would be executed. Similarly, we find that $\beta \leq a, \overline b$. Since the loop is exited after two iterations, we must have $\gamma \leq \overline a$. Hence, we find
	\[
	\alpha p \beta q \gamma \in \int{(p +_b q)^{(a)}} \Leftrightarrow \alpha \leq a, b \text{ and } \beta \leq a, \overline{b} \text{ and } \gamma \leq \overline{a}. 
	\]
	\end{example}
	We state two simple facts that will be useful later on. 
	\begin{lemma}
    For any two languages $L, K$ of guarded strings, and primitive program $p$, we have:
    \begin{enumerate}[label = (\roman*)]
        \item $L^{n + 1} = L \fusion L^n$; \ \ \ \ (ii) $\int{p} \fusion L = \int{p} \fusion K$ implies $L = K$. 
    \end{enumerate}
	\end{lemma}
	\begin{remark}
	\label{rmk:determinacy}
	The fact that $\gkat$ models deterministic programs is reflected in the fact that sets of guarded strings arising as interpretations of $\gkat$-expressions satisfy a certain \emph{determinacy property}\index{determinacy property}. Namely, for every $x \alpha y$ and $x \alpha z$ in $L$, either $y$ and $z$ are both empty, or both begin with the same primitive program. We refer the reader to~\cite{Smolka20} for more details.
	\end{remark}
	\begin{remark}
	The language semantics of $\gkat$ is the same as that of $\kat$ (see~\cite{KozenKAT97}), in the sense that $\int{e} = \int{\varphi(e)}$, where $\varphi$ is the embedding from Remark \ref{rem:katgkatsyn}, the semantic brackets on the right-hand side denote the standard interpretation in $\kat$, and $e$ is any $\gkat$-expression.
	\end{remark}
	\section{The non-well-founded proof system $\mathsf{SGKAT}^\infty$}
	\label{sec:thesgkatsystem}
	In this section we commence our proof-theoretical study of $\gkat$. We will present a cyclic sequent system for $\gkat$, inspired by the cyclic sequent system for Kleene Algebra presented in~\cite{DasPous17}. In passing, we will compare our system to the latter.
	\begin{definition}[Sequent]
	\label{defn:sequent}
	A \emph{sequent}\index{sequent} is a triple $(\Gamma, A, \Delta)$, written $\Gamma \Rightarrow_A \Delta$, where $A \subseteq \mathsf{At}$ and $\Gamma$ and $\Delta$ are (possibly empty) lists of $\gkat$-expressions.
	\end{definition}
	 The list on the left-hand side of a sequent is called its \emph{antecedent}\index{antecendent}, and the list on the right-hand side its \emph{succedent}\index{succedent}. In general we refer to lists of expressions as \emph{cedents}.  The symbol $\epsilon$ refers to the empty cedent. 
     \remark{As the system in~\cite{DasPous17} only deals with Kleene Algebra, it does not include tests. We choose the deal with the tests present in $\gkat$ by augmenting each sequent by a set of atoms. This tucks away the Boolean content, as is usual in the $\gkat$ literature, allowing us to omit propositional rules.
    \begin{definition}[Validity]
	 We say that a sequent $e_1, \ldots, e_n \Rightarrow_A f_1, \ldots, f_m$ is \emph{valid}\index{valid} whenever $A \diamond \llbracket e_1 \cdots e_n \rrbracket \subseteq \llbracket f_1 \cdots f_n \rrbracket$. 
	 \end{definition}
	 We often abuse notation writing $\int{\Gamma}$ instead of $\int{e_1 \cdots e_n}$, where $\Gamma = e_1, \ldots, e_n$. 
	 	 	 {
	 \FrameSep0pt
		 \begin{framed}
		 		\begin{enumerate}[align = left, leftmargin=1.3cm, itemsep = 0.5cm]
			\item[\textbf{Left logical rules}]
			\begin{align*}
				\AxiomC{$\Gamma \Rightarrow_{A \restriction b} \Delta$}
				\RightLabel{\footnotesize $b$-$l$}
				\UnaryInfC{$b, \Gamma \Rightarrow_A \Delta$}
				\DisplayProof
				&&
				\AxiomC{$e ,\Gamma \Rightarrow_{A \restriction b} \Delta$}
				\AxiomC{$f ,\Gamma \Rightarrow_{A \restriction \overline b} \Delta$}
				\RightLabel{\footnotesize $+_b$-$l$}
				\BinaryInfC{$e +_b f, \Gamma \Rightarrow_A \Delta$}
				\DisplayProof
			\end{align*}
			\begin{align*}
				\AxiomC{$e, g, \Gamma \Rightarrow_A \Delta$}
				\RightLabel{\footnotesize $\cdot$-$l$}
				\UnaryInfC{$e \cdot g, \Gamma \Rightarrow_A \Delta$}
				\DisplayProof
				&&
				\AxiomC{$e , e^{(b)},\Gamma \Rightarrow_{A \restriction b} \Delta$}
				\AxiomC{$\Gamma \Rightarrow_{A \restriction \overline b} \Delta$}
				\RightLabel{\footnotesize $(b)$-$l$}
				\BinaryInfC{$e^{(b)}, \Gamma \Rightarrow_A \Delta$}
				\DisplayProof
			\end{align*}
			\item[\textbf{Right logical rules}]
				\begin{align*}
				\AxiomC{$\Gamma \Rightarrow_A \Delta$}
				\RightLabel{\footnotesize $b$-$r$}
				\LeftLabel{$(\dagger)$}
				\UnaryInfC{$\Gamma \Rightarrow_A b, \Delta$}
				\DisplayProof
				&&
				\AxiomC{$\Gamma \Rightarrow_{A \restriction b} e, \Delta$}
				\AxiomC{$\Gamma \Rightarrow_{A \restriction \overline b} f, \Delta$}
				\RightLabel{\footnotesize $+_b$-$r$}
				\BinaryInfC{$\Gamma \Rightarrow_A e +_b f, \Delta$}
				\DisplayProof
			\end{align*}
			\begin{align*}
				\AxiomC{$\Gamma \Rightarrow_A   e, f, \Delta$}
				\RightLabel{\footnotesize $\cdot$-$r$}
				\UnaryInfC{$\Gamma \Rightarrow_A e \cdot f, \Delta$}
				\DisplayProof
				&&
				\AxiomC{$\Gamma \Rightarrow_{A \restriction b} e, e^{(b)}, \Delta$}
				\AxiomC{$\Gamma \Rightarrow_{A \restriction \overline b} \Delta$}
				\RightLabel{\footnotesize $(b)$-$r$}
				\BinaryInfC{$\Gamma \Rightarrow_A e^{(b)}, \Delta$}
				\DisplayProof
			\end{align*}
			\item[\textbf{Axioms and modal rules}]
			\begin{align*}
				\AxiomC{}
				\RightLabel{\footnotesize $\mathsf{id}$}
				\UnaryInfC{$\epsilon \Rightarrow_A \epsilon$}
				\DisplayProof
				&&
				\AxiomC{}
				\RightLabel{\footnotesize $\mathsf{\bot}$}
				\UnaryInfC{$\Gamma \Rightarrow_\emptyset \Delta$}
				\DisplayProof
				&&
				\AxiomC{$\Gamma \Rightarrow_\mathsf{At} \Delta$}
				\RightLabel{\footnotesize $\mathsf{k}$}
				\UnaryInfC{$p, \Gamma \Rightarrow_A p, \Delta$}
				\DisplayProof
				&&
				\AxiomC{$\Gamma \Rightarrow_\mathsf{At} 0$}
				\RightLabel{\footnotesize $\mathsf{k}_0$}
				\UnaryInfC{$p, \Gamma \Rightarrow_A \Delta$}
				\DisplayProof
			\end{align*}
		\end{enumerate}
		\end{framed}
	\vspace{-10pt}
		\captionof{figure}{The rules of $\mathsf{SGKAT}$. The side condition $(\dagger)$ requires that $A \restriction b = A$.}
		\label{fig:sgkat}
	}
		 \begin{example}
	 \label{exmp:validsgkatsequent}
	 An example of a valid sequent is given by $(cp)^{(b)} \Rightarrow_\mathsf{At} (p (cp +_b 1))^{(b)}$. The antecedent denotes guarded strings $\alpha_1 p \alpha_2 p \cdots \alpha_n p \alpha_{n+1}$ where $\alpha_i \leq b, c$ for each $1 \leq i \leq n$, and $\alpha_{n+1} \leq \overline b$. The succedent denotes such strings where $\alpha_{i} \leq c$ is only required for those $1 \leq i \leq n$ where $i$ is even.
	 \end{example}
	 \begin{remark}
	 Like the sequents for Kleene Algebra in~\cite{DasPous17}, our sequents express language \emph{inclusion}, rather than language equivalence. For Kleene Algebra this difference is insignificant, as the two notions are interdefinable using unrestricted union: $\int{e} \subseteq \int{f} \Leftrightarrow \int{e + f}  = \int{f}$. For $\gkat$, however, it is not clear how to define language inclusion in terms of language equivalence. As a result, an advantage of axiomatising language inclusion rather than language equivalence, is that the while-operator can be axiomatised as a \emph{least} fixed point, eliminating the need for a \emph{strict productivity} requirement as is present in the axiomatisation in~\cite{Smolka20}. 
	 \end{remark}
	 Given a set of atoms $A$ and a test $b$, we write $A \restriction b$ for $A \fusion \int{b}$, \emph{i.e.} the set of atoms $\{\alpha \in A : \alpha \leq b\}$. The rules of $\mathsf{SGKAT}$\index{$\mathsf{SGKAT}$} are given in Figure \ref{fig:sgkat}. Importantly, the rules are always applied to the leftmost expression in a cedent. As a result, we have the following lemma, that later will be used in the completeness proof. 
    \begin{lemma}
    \label{lem:atmostone}
    Let $\Gamma \Rightarrow_A \Delta$ be a sequent, and let $\mathsf{r}$ be any rule of $\mathsf{SGKAT}$. Then there is at most one rule instance of $\mathsf{r}$ with conclusion $\Gamma \Rightarrow_A \Delta$.
    \end{lemma}
    
	 \begin{remark}
	 Following~\cite{DasPous17}, we call $\mathsf{k}$ a `modal' rule. The reason is simply that it looks like the rule $\mathsf{k}$ (sometimes called $\mathsf{K}$ or $\lb$) in the standard sequent calculus for basic modal logic. Our system also features a second modal rule, called $\mathsf{k}_0$. Like $\mathsf{k}$, this rule adds a primitive program $p$ to the antecedent of the sequent. Since the premiss of $\mathsf{k_0}$ entails that $\int{\Gamma} = \int{0}$, the antecedent of its conclusion will denote the empty language, and is therefore included in any succedent $\Delta$.  
	\end{remark}
	\begin{remark}
	\label{rmk:symmetry}
	Note that the rules of $\mathsf{SGKAT}$ are highly symmetric. Indeed, the only rules that behave differently on the left than on the right, are the $b$-rules and $\mathsf{k}_0$. Note that $b$-$l$ changes the set of atoms, while $b$-$r$ uses a side condition. The asymmetry of $\mathsf{k}_0$ is clear: the succedent of the premiss has a $0$, whereas the antecedent does not. A third asymmetry will be introduced in Definition \ref{def:sgkatproofsfairness}, with a condition on infinite branches that is sensitive to $(b)$-$l$ but not to $(b)$-$r$. 
\end{remark}
	\begin{remark}
	\label{rmk:earlyterm}
	The authors of~\cite{SchmidKK021} study a variant of $\gkat$ that omits the so-called \emph{early termination axiom}, which equates all programs that eventually fail. They give a denotational model of this variant in the form of certain kinds of trees. We conjecture that omitting the rule $\mathsf{k}_0$ from our system will make it sound and complete with respect to this denotational model. 
\end{remark}
An \emph{$\mathsf{SGKAT}^\infty$-derivation}\index{derivation!$\mathsf{SGKAT}^\infty$-} is a (possibly infinite) tree generated by the rules of $\mathsf{SGKAT}$. Such a derivation is said to be \emph{closed} if every leaf is an axiom.
\begin{definition}[Proof]
\label{def:sgkatproofsfairness}
A closed $\mathsf{SGKAT}^\infty$-derivation is said to be an \emph{$\mathsf{SGKAT}^\infty$-proof}\index{proof!$\mathsf{SGKAT}^\infty$-} if every infinite branch is \emph{fair}\index{proof!$\mathsf{SGKAT}^\infty$-!fair} for $(b)$-$l$, \emph{i.e.} contains infinitely many applications of the rule $(b)$-$l$.  
\end{definition}
We write $\mathsf{SGKAT} \vdash^\infty \Gamma \Rightarrow_A \Delta$ if there is an $\mathsf{SGKAT}^\infty$-proof of $\Gamma \Rightarrow_A \Delta$. 
\newpage
\begin{example}
    Not every $\mathsf{SGKAT}^\infty$-derivation is a proof. Consider for instance the following derivation, where $(\bullet)$ indicates that the derivation repeat itself.
    	{
	 \FrameSep0pt
		 \begin{framed}
		\begin{align*}
        \AxiomC{$p \Rightarrow_{\mathsf{At} \restriction b} 1, 1^{(b)}, p$ \ $(\bullet)$}
        \AxiomC{}
        \RightLabel{\footnotesize $\mathsf{\bot}$}
        \UnaryInfC{$p \Rightarrow_{\emptyset} p$}
        \LeftLabel{\footnotesize $(b)$-$r$}
        \BinaryInfC{$p \Rightarrow_{\mathsf{At} \restriction b} 1^{(b)}, p$}
        \LeftLabel{\footnotesize $1$-$r$}
        \UnaryInfC{$p \Rightarrow_{\mathsf{At} \restriction b} 1, 1^{(b)}, p$ \ $(\bullet)$}
        \AxiomC{}
        \RightLabel{\footnotesize $\mathsf{id}$}
        \UnaryInfC{$\epsilon \Rightarrow_{\mathsf{At}} \epsilon$}
        \RightLabel{\footnotesize $\mathsf{k}$}
        \UnaryInfC{$p \Rightarrow_{\mathsf{At} \restriction \overline{b}} p$}
        \RightLabel{\footnotesize $(b)$-$r$}
        \BinaryInfC{$p \Rightarrow_{\mathsf{At}} 1^{(b)}, p$}
  		\RightLabel{\footnotesize $\cdot$-$r$}
		\UnaryInfC{$p \Rightarrow_{\mathsf{At}} 1^{(b)} \cdot p$}
        \DisplayProof
    	\end{align*}
	\end{framed}
	\vspace{-10pt}
		\captionof{figure}{An $\mathsf{SGKAT}^\infty$-derivation that is not a proof.\\[10pt]}
		\label{fig:derivationthatisnotaproof}
	}
\end{example}
\begin{example}
	\label{example:sgkatproof}
	Let $\Delta_1 := (p (cp +_b 1))^{(b)}$ and $\Delta_2 := cp   +_b 1, \Delta_1$. The following proof $\Pi_1$ is an example $\mathsf{SGKAT}^\infty$-proof of the sequent of Example \ref{exmp:validsgkatsequent}. We again use $(\bullet)$ to indicate that the proof repeats itself at this leaf and, for the sake of readability, omit branches that can be closed immediately by an application of $\bot$.
	{
	 \FrameSep0pt
		 \begin{framed}
		\begin{align*}
		\small
		\AxiomC{$ (c p)^{(b)} \Rightarrow_{\mathsf{At}} \Delta_1$ \ $(\bullet)$}
		\LeftLabel{\footnotesize $\mathsf{k}$}
		\UnaryInfC{$p, (cp)^{(b)} \Rightarrow_{\mathsf{At} \restriction{bc}} p, \Delta_1$}
		\LeftLabel{\footnotesize $c$-$r$}
		\UnaryInfC{$p, (cp)^{(b)} \Rightarrow_{\mathsf{At} \restriction bc} c, p, \Delta_1$}
		\LeftLabel{\footnotesize $\cdot$-$r$}
		\UnaryInfC{$p, (cp)^{(b)} \Rightarrow_{\mathsf{At} \restriction bc} cp, \Delta_1$}
		\LeftLabel{\footnotesize $+_b$-$r$}
		\UnaryInfC{$p, (cp)^{(b)} \Rightarrow_{\mathsf{At} \restriction bc} \Delta_2$}
		\LeftLabel{\footnotesize $c$-$l$}
		\UnaryInfC{$c, p, (cp)^{(b)} \Rightarrow_{\mathsf{At} \restriction b} \Delta_2$}
		\LeftLabel{\footnotesize $\cdot$-$l$}
		\UnaryInfC{$cp, (cp)^{(b)} \Rightarrow_{\mathsf{At} \restriction b} \Delta_2$}
		\AxiomC{}
		\RightLabel{\footnotesize $\mathsf{id}$}
		\UnaryInfC{$\epsilon \Rightarrow_{\mathsf{At} \restriction \overline b} \epsilon$}
		\RightLabel{\footnotesize $(b)$-$r$}
		\UnaryInfC{$\epsilon \Rightarrow_{\mathsf{At} \restriction \overline b} \Delta_1$}
		\RightLabel{\footnotesize $1$-$r$}
		\UnaryInfC{$\epsilon \Rightarrow_{\mathsf{At} \restriction \overline b} 1, \Delta_1$}
		\RightLabel{\footnotesize $+_b$-$r$}
		\UnaryInfC{$\epsilon \Rightarrow_{\mathsf{At} \restriction \overline b} \Delta_2$}
		\RightLabel{\footnotesize $(b)$-$l$}
		\BinaryInfC{$(cp)^{(b)} \Rightarrow_\mathsf{At} \Delta_2$}
		\LeftLabel{\footnotesize $\mathsf{k}$}
		\UnaryInfC{$p, (c p)^{(b)}  \Rightarrow_{\mathsf{At} \restriction bc} p, (cp +_b 1), \Delta_1$}
		\LeftLabel{\footnotesize $\cdot$-$r$}
		\UnaryInfC{$p, (cp)^{(b)} \Rightarrow_{\mathsf{At} \restriction bc}  p (cp +_b 1), \Delta_1$}
		\LeftLabel{\footnotesize $(b)$-$r$}
		\UnaryInfC{$p, (cp)^{(b)} \Rightarrow_{\mathsf{At} \restriction bc} \Delta_1$}
		\LeftLabel{$c$-$l$}
		\UnaryInfC{$c, p, (cp)^{(b)} \Rightarrow_{\mathsf{At} \restriction b} \Delta_1$}
		\LeftLabel{$\cdot$-$l$}
		\UnaryInfC{$cp, (cp)^{(b)} \Rightarrow_{\mathsf{At} \restriction b} \Delta_1$}
		\AxiomC{}
		\RightLabel{\footnotesize $\mathsf{id}$}
		\UnaryInfC{$\epsilon \Rightarrow_{\mathsf{At} \restriction \overline b} \epsilon$}
		\RightLabel{\footnotesize $(b)$-$r$}
		\UnaryInfC{$\epsilon \Rightarrow_{\mathsf{At} \restriction \overline b} \Delta_1$}
		\RightLabel{\footnotesize $(b)$-$l$}
		\BinaryInfC{$(cp)^{(b)} \Rightarrow_{\mathsf{At}} \Delta_1$ \ $(\bullet)$}
		\DisplayProof
	\end{align*}
	\end{framed}
	\vspace{-10pt}
		\captionof{figure}{The $\mathsf{SGKAT}^\infty$-proof $\Pi_1$.\\[10pt]}
		\label{fig:proofp1}
	}

	To illustrate the omission of branches that can be immediately closed by an application of $\bot$, let us write out the two applications of $+_b$-$r$ in $\Pi_1$.
	\begin{align*}
	\AxiomC{$\epsilon \Rightarrow_{\mathsf{At} \restriction bc} cp, \Delta_1$}
	\AxiomC{}
	\RightLabel{\footnotesize $\bot$}
	\UnaryInfC{$\epsilon \Rightarrow_{\emptyset} 1, 	\Delta_1$}
	\RightLabel{\footnotesize $+_b$-$r$}
	\BinaryInfC{$\epsilon \Rightarrow_{\mathsf{At} \restriction bc} \Delta_2$}
	\DisplayProof
	&&
	\AxiomC{}
	\LeftLabel{\footnotesize $\bot$}
	\UnaryInfC{$\epsilon \Rightarrow_{\emptyset} cp, \Delta_1$}
	\AxiomC{$\epsilon \Rightarrow_{\mathsf{At} \restriction \overline b} 1, \Delta_1$}
	\RightLabel{\footnotesize $+_b$-$r$}
	\BinaryInfC{$\epsilon \Rightarrow_{\mathsf{At} \restriction \overline b} \Delta_2$}
	\DisplayProof
	\end{align*}
	It can also be helpful to think of the set of atoms as \emph{selecting} one of the premisses.
\end{example}
We close this section with a useful definition and a lemma.
\begin{definition}[Exposure]
A list $\Gamma$ of expressions is said to be \emph{exposed}\index{exposed} if it is either empty or begins with a primitive program. 
\end{definition}
Recall that the sets of primitive tests and primitive programs are disjoint. Hence an exposed list $\Gamma$ cannot start with a test. The following easy lemma will be useful later on.
\begin{lemma}
\label{lem:exposed}
Let $\Gamma$ and $\Delta$ be exposed lists of expressions. Then:
\begin{enumerate}[label = (\roman*)]
\item $\alpha x \in \int{\Gamma} \Leftrightarrow \beta x \in \int{\Gamma}$ for all $\alpha, \beta \in \mathsf{At}$
\item $\Gamma \Rightarrow_{\mathsf{At}} \Delta$ is valid if and only if $\Gamma \Rightarrow_A \Delta$ is valid for some $A \not= \emptyset$. 
\end{enumerate}
\end{lemma}
\section{Soundness} 
\label{sec:sgkatsoundess}
In this section we prove that $\mathsf{SGKAT}^\infty$ is sound. We will first prove that \emph{well-founded} (that is, finite) $\mathsf{SGKAT}^\infty$-proofs are sound. The following straightforward facts will be useful in the soundness proof.
\begin{lemma}
\label{lem:soundnessaux}
For any set $A$ of atoms, test $b$, and cedent $\Theta$, we have:
\begin{enumerate}[label = (\roman*)]
	\item $\int{e +_b f, \Theta} = (\int{b} \fusion \int{e, \Theta}) \cup (\int{\overline b} \fusion \int{f, \Theta})$;
	\item $\int{e^{(b)}, \Theta} = (\int{b} \fusion \int{e, e^{(b)}, \Theta}) \cup (\int{\overline{b}} \fusion \int{\Theta})$. 
\end{enumerate}
\end{lemma}
We prioritise the rules of $\mathsf{SGKAT}$ in order of occurrence in Figure \ref{fig:sgkat}, reading left-to-right, top-to-bottom. Hence, each left logical rule is of higher priority than each right logical rule, which is of higher priority than each axiom or modal rule. 
Recall that a rule is \emph{sound} if the validity of all its premisses implies the validity of its conclusion. Conversely, a rule is \emph{invertible}\index{invertible} if the validity of its conclusion implies the validity of all of its premisses. 

We say that a rule application \emph{has priority} of there is no higher-priority rule with the same conclusion. Conveniently, the following proposition entails that every rule instance which has priority is invertible. This will aid our proof search procedure in Section \ref{sec:sgkatcompleteness} .
\begin{proposition}
	\label{prop:wellfoundedsoundness}
	Every rule of $\mathsf{SGKAT}$ is sound. Moreover, every rule is invertible except for $\mathsf{k}$ and $\mathsf{k}_0$, which are invertible whenever they have priority. 
\end{proposition}
\begin{proof}[sketch]
We treat two illustrative cases. For the rule $+_b$-$r$, we find 
\begin{align*}
	&A \fusion \int{\Gamma} \subseteq \int{e +_b f} \fusion \int{\Delta} \\\
	&\ \  \Leftrightarrow A \fusion \int{\Gamma} \subseteq (\int{b} \fusion \int{e, \Delta}) \cup (\int{\overline b} \fusion \int{f, \Delta}) \\
	&\ \ \ \ \Leftrightarrow A \restriction b \fusion \int{\Gamma} \subseteq \int{e, \Delta} \text{ or } A \restriction \overline b \subseteq \int{f, \Delta},
	\end{align*}
	where the first equivalence holds due to Lemma \ref{lem:soundnessaux}.(ii), and the second due to $A \fusion \int{\Gamma} = (\int{b} \fusion A \fusion \int{\Gamma}) \cup (\int{\overline b} \fusion A \fusion \int{\Gamma})$ and Lemma \ref{lem:soundnessaux}.(i).
 
  The other rule we will treat is $\mathsf{k}$. Suppose first that some application of $\mathsf{k}$ does \emph{not} have priority. The only rule of higher priority than $\mathsf{k}$ which can have a conclusion of the form $p, \Gamma \Rightarrow_A p, \Delta$ is $\bot$. In this case $A = \emptyset$, which means that the conlusion must be valid. Hence any application of $\mathsf{k}$ that does not have priority is vacuously sound. It need, however, not be invertible, as the following rule instance demonstrates
	\[
	\AxiomC{$1 \Rightarrow_\mathsf{At} 0$}
	\LeftLabel{$\mathsf{k}$}
	\UnaryInfC{$p, 1 \Rightarrow_\emptyset p, 0$}
	\DisplayProof
	\]

	Next, suppose that some application of $\mathsf{k}$ does have priority. This means that the set $A$ of atoms in the conclusion $p, \Gamma \Rightarrow_A p, \Delta$ is \emph{not} empty. We will show that under this restriction the rule is both sound and invertible. Let $\alpha \in A$. We have
	\begin{align*}
	A \fusion \int{p, \Gamma} \subseteq \int{p, \Delta} &\Leftrightarrow A \fusion \int{p} \fusion \int{\Gamma} \subseteq \int{p} \fusion \int{\Delta} &\text{(seq. int.)}\\
	&\Leftrightarrow {\alpha} \fusion \int{p} \fusion \int{\Gamma} \subseteq \int{p} \fusion \int{\Delta} &\text{($\alpha \in A$, Lem. \ref{lem:exposed})}\\
	&\Leftrightarrow \int{p} \fusion \int{\Gamma} \subseteq \int{p} \fusion \int{\Delta} &\text{(Lem. \ref{lem:exposed})} \\
	&\Leftrightarrow \int{\Gamma} \subseteq \int{\Delta}, &\text{($\dagger$)}
	\end{align*}
	as required. The step marked by $\dagger$ is the following property of guarded languages: $\int{p} \fusion L = \int{p} \fusion K$ implies $L = K$.
\end{proof}
Proposition \ref{prop:wellfoundedsoundness} entails that all finite proofs are sound. We will now extend this result to non-well-founded proofs, closely following the treatment in~\cite{DasPous17}. We first recursively define a syntactic abbreviation: $[e^{(b)}]^0 := \overline b$ and $[e^{(b)}]^{n + 1} := be[e^{(b)}]^n$. 
\begin{lemma}
	\label{lem:finapprox}
	For every $n \in \mathbb{N}$: if we have $\mathsf{SGKAT} \vdash^\infty e^{(b)}, \Gamma \Rightarrow_A \Delta$, then we also have $\mathsf{SGKAT} \vdash^\infty [e^{(b)}]^n, \Gamma \Rightarrow_A \Delta$.
\end{lemma}
We let the \emph{while-height}\index{while-height} $\wh(e)$\index{$\wh$} be the maximal nesting of while loops in a given expression $e$. Formally, 
\begin{itemize}
	\item $\wh(b) = \wh(p) = 0;$ \hfill $-$ $\wh(e \cdot f) = \wh(e +_b f) = \max\{\wh(e), \wh(f)\};$
	\item $\wh(e^{(b)}) = \wh(e) + 1.$
\end{itemize}
Given a list $\Gamma$, the \emph{weighted while-height}\index{weighted while-height} $\wwh(\Gamma)$\index{$\wwh$} of $\Gamma$ is defined to be the multiset $[\wh(e) : e \in \Gamma]$. We order such multisets using the Dershowitz–Manna ordering (for linear orders): we say that $N < M$ if and only if $N \not= M$ and for the greatest $n$ such that $N(n) \not= M(n)$, it holds that $N(n) < M(n)$. 

Note that in any $\mathsf{SGKAT}$-derivation the weighted while-height of the antecedent does not increase when reading bottom-up. Moreover, we have:
\begin{lemma}
	\label{lem:strictdescdm}
	$\wwh([e^{(b)}]^n, \Gamma) < \wwh(e^{(b)}, \Gamma)$ for every $n \in \mathbb{N}$. 
\end{lemma}
Finally, we can prove the soundness theorem using induction on $\wwh(\Gamma)$. 
\begin{theorem}[Soundness]
	\label{thm:soundness}
	If $\mathsf{SGKAT} \vdash^\infty \Gamma \Rightarrow_A \Delta$, then $A \diamond \int{\Gamma} \subseteq \int{\Delta}$.
\end{theorem}
\begin{proof}
	We prove this by induction on $\wwh(\Gamma)$. Given a proof $\pi$ of $\Gamma \Rightarrow_A \Delta$, let $\mathcal{B}$ contain for each infinite branch of $\pi$ the node of least depth to which a rule $(b)$-$l$ is applied. Note that $\mathcal{B}$ must be finite, for otherwise, by K\H{o}nig's Lemma, the proof $\pi$ cut off along $\mathcal{B}$ would have an infinite branch that does not satisfy the fairness condition. 

	Note that Proposition~\ref{prop:wellfoundedsoundness} entails that of every finite derivation with valid leaves the conclusion is valid. Hence, it suffices to show that each of the nodes in $\mathcal{B}$ is valid. To that end, consider an arbitrary such node labelled $e^{(b)}, \Gamma' \Rightarrow_{A'} \Delta'$ and the subproof $\pi'$ it generates. By Lemma \ref{lem:finapprox}, we have that $[e^{(b)}]^{n}, \Gamma' \Rightarrow_{A'} \Delta'$ is provable for every $n$. Lemma \ref{lem:strictdescdm} gives $\wwh([e^{(b)}]^{n}, \Gamma') < \wwh(e^{(b)}, \Gamma') \leq \wwh(\Gamma)$, and thus we may apply the induction hypothesis to obtain
	\[
	A' \diamond \llbracket [e^{(b)}]^{n} \rrbracket \diamond \llbracket \Gamma \rrbracket \subseteq \llbracket \Delta \rrbracket
	\]
	for every $n \in \mathbb{N}$. Then by 
	\begin{align*}
		\bigcup_{n}  (A' \diamond \llbracket [e^{(b)}]^{n} \rrbracket \diamond \llbracket \Gamma \rrbracket) = A' \diamond	\bigcup_{n} (\llbracket [e^{(b)}]^{n} \rrbracket)  \diamond \llbracket \Gamma \rrbracket = A' \diamond \llbracket e \rrbracket^{\llbracket b \rrbracket} \diamond \llbracket \Gamma \rrbracket,
	\end{align*}
	we obtain that $e^{(b)}, \Gamma' \Rightarrow_{A'} \Delta'$ is valid, as required.
\end{proof}
\section{Regularity}
\label{sec:finite-stateness}
Before we show that $\mathsf{SGKAT}^\infty$ is not only sound, but also complete, we will first show that every $\mathsf{SGKAT}^\infty$-proof is \emph{finite-state}, \emph{i.e.} that it contains at most finitely many distinct sequents. 

The results of this section crucially depend on the fact that we are only applying rules to the leftmost expressions of cedents. Indeed, otherwise one could easily create infinitely many distinct sequents by simply unravelling the same while loop $e^{(b)}$ infinitely often. 

Our treatment differs from that in~\cite{DasPous17} in two major ways. First, we formalise the notion of (sub)occurrence using the standard notion of a \emph{syntax tree}. Secondly, and more importantly, we obtain a quadratic bound on the number of distinct sequents occurring in a proof, rather than an exponential one. In fact, we will show that the number of distinct antecedents (succedents) is \emph{linear} in the size of the syntax tree of the antecedent (succedent) of the root. We will do this by showing that each leftmost expression of a cedent in the proof (given as node of the syntax tree of a root cedent) can only occur in the proof as the leftmost expresson of that \emph{unique} cedent. 
\begin{definition}
The \emph{syntax tree} $(T_e, l_e)$\index{syntax tree}\index{$T_e$} of an expression $e$ is a well-founded, labelled and ordered tree, defined by the following induction on $e$. 
\begin{itemize}
	\item If $e$ is a test or primitive program, its syntax tree only has a root node $\rho$, with label $l_e(\rho) := e$. 
	\item If $e = f_1 \circ f_2$ where $\circ = \cdot$ or $\circ = +_b$, its syntax tree again has a root node $\rho$ with label $l_e(\rho) = e$, and with two outgoing edges. The first edge connects $\rho$ to $(T_{f_1}, l_{f_1})$, the second edge connects it to $(T_{f_2}, l_{f_2})$.
	\item If $e = f^{(b)}$, its syntax tree again has a root node $\rho$ with label $l_e(\rho) = e$, but now with just one outgoing edge. This edge connects $\rho$ to $(T_{f}, l_f)$. 
\end{itemize}
\end{definition}
\begin{definition}
An \emph{$e$-cedent} is a list of nodes in the syntax tree of $e$. The \emph{realisation} of an $e$-cedent $u_1, \ldots, u_n$ is the cedent $l_e(u_1), \ldots, l_e(u_n)$. 
\end{definition}
Given the leftmost expression of a cedent, we will now explicitly define the cedent that it must be the leftmost expression of.
\begin{definition}
Let $u$ be a node in the syntax tree of $e$. We define the $e$-cedent $\mathsf{tail}(u)$ inductively as follows: 
\begin{itemize}
	\item For the root $\rho$ of $T_e$, we set $\mathsf{tail}(\rho)$ to be the empty list $\epsilon$.
	\item For every node $u$ of $T_e$, we define $\mathsf{tail}$ on its children by a case distinction on the main connective $\mathsf{mc}$ of $u$:
	\begin{itemize}
		\item if $\mathsf{mc} = \cdot$, let $u_1$ and $u_2$ be, respectively, the first and second child of $u$. We set $\mathsf{tail}(u_1) := u_2, \mathsf{tail}(u)$ and $\mathsf{tail}(u_2) := \mathsf{tail}(u)$. 
		\item if $\mathsf{mc} = +_b$, let $u_1$ and $u_2$ again be its first and second child. We set $\mathsf{tail}(u_1) := \mathsf{tail}(u_2) := \mathsf{tail}(u)$.
		\item if $\mathsf{mc} = (-)^{(b)}$, let $v$ be the single child of $u$. We set $\mathsf{tail}(v) := u, \mathsf{tail}(u)$.
	\end{itemize}
\end{itemize}
An $e$-cedent is called $\mathsf{tail}$-generated if it is empty or of the form $u, \mathsf{tail}(u)$ for some node $u$ in the syntax tree of $e$. 
\end{definition}
\begin{example}
Below is the syntax tree of $(p(p +_b 1))^{(b)}$ and a calculation of $\mathsf{tail}(u_3)$. \\[0.4cm]
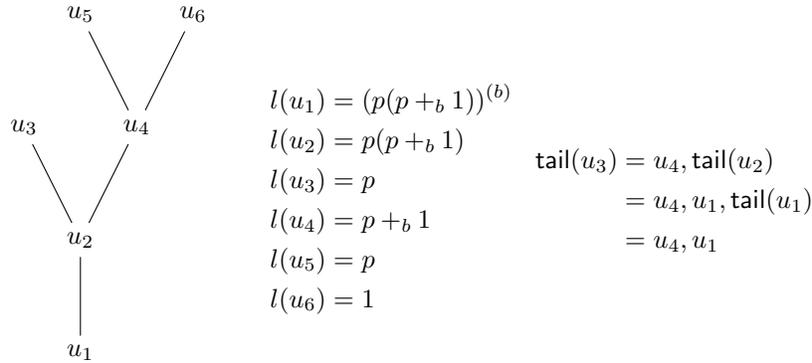
\begin{minipage}{\linewidth/3}
	\begin{center}
		\begin{tikzpicture}
			\node {$u_1$} [grow'=up]
			child {node {$u_2$} 
				child {node {$u_3$}}
				child {node {$u_4$}
				child {node {$u_5$}}
				child {node {$u_6$}}}};
		\end{tikzpicture}
	\end{center}
\end{minipage}
\begin{minipage}{\linewidth/4}
\begin{align*}
l(u_1) &= (p(p +_b 1))^{(b)} \\
l(u_2) &= p(p +_b 1) \\
l(u_3) &= p \\
l(u_4) &= p +_b 1 \\
l(u_5) &= p \\
l(u_6) &= 1 
\end{align*}
\end{minipage}
\begin{minipage}{\linewidth/3}
\begin{align*}
\mathsf{tail}(u_3) &= u_4, \mathsf{tail}(u_2) \\
				   &= u_4, u_1, \mathsf{tail}(u_1) \\
				   &= u_4, u_1
\end{align*}
\end{minipage}
\end{example}
The following lemma embodies the key idea for the main result of this section: every leftmost expression is the leftmost expression of a unique cedent.
\begin{lemma}
\label{lem:realisationoftailgenerated}
Let $\pi$ be an $\mathsf{SGKAT^\infty}$-derivation of a sequent of the form $e \Rightarrow_A f$. Then every antecedent in $\pi$ is the realisation of a $\mathsf{tail}$-generated $e$-sequent, and every succedent is the realisation of a $\mathsf{tail}$-generated $f$-sequent or $0$-sequent.
\end{lemma}
\begin{proof}
    We first prove the following claim.
\begin{quote}
Let $e$ be an expression and let $u$ be a node in its syntax tree. Then $\mathsf{tail}(u)$ is a $\mathsf{tail}$-generated $e$-sequent. 
\end{quote}
We prove this by induction on the syntax tree of $e$. For the root $\rho$, we have $\mathsf{tail}(\rho) = \epsilon$, which is $\mathsf{tail}$-generated by definition. Now suppose that the thesis holds for some arbitrary node $u$ in the syntax tree of $e$. We will show that the thesis holds for the children of $u$ by a case distinction on the main connective $\mathsf{mc}$ of $u$.
\begin{itemize}
	\item $\mathsf{mc} = \cdot$. Let $u_1$ and $u_2$ be the first and second child of $u$, respectively. We have $\mathsf{tail}(u_1) = u_2, \mathsf{tail}(u) = u_2, \mathsf{tail}(u_2)$, which is $\mathsf{tail}$-generated by definition. Moreover, we have that $\mathsf{tail}(u_2) = \mathsf{tail}(u)$ is $\mathsf{tail}$-generated by the induction hypothesis.

	\item $\mathsf{mc} = +_b$. Then for each child $v$ of $u$, we have $\mathsf{tail}(v) = \mathsf{tail}(u)$ and thus we can again invoke the induction hypothesis. 

	\item $\mathsf{mc} = (-)^{(b)}$. Then for the single child $v$ of $u$, it holds that $\mathsf{tail}(v) = u, \mathsf{tail}(u)$, which is $\mathsf{tail}$-generated by definition. 
\end{itemize}
Using this claim, the lemma follows by bottom-up induction on $\pi$. For the base case, note that $e$ and $f$ are realisations of the roots of their respective syntax trees. Such a root $\rho$ is $\mathsf{tail}$-generated, since $\rho = \rho, \epsilon = \rho, \mathsf{tail}(\rho)$. The induction step follows by direct inspection of the rules of $\mathsf{SGKAT}$. 
\end{proof}
The number of realisations of $\mathsf{tail}$-generated $e$-sequents is clearly linear in the size of the syntax tree of $e$, for every expression $e$. Hence we obtain:
\begin{corollary}
	\label{cor:anyfrugal}
	The number of distinct sequents in an $\mathsf{SGKAT}^\infty$-proof of $e \Rightarrow_A f$ is quadratic in $|T_e| + |T_f|$. 
\end{corollary}
Note that the above lemma and corollary can easily be generalised to arbitrary (rather than singleton) cedents, by rewriting each cedent $e_1, \ldots, e_n$ as $e_1 \cdots e_n$. 

Recall that a non-well-founded tree is \emph{regular} if it contains only finitely many pairwise non-isomorphic subtrees. The following corollary follows by a standard argument in the literature (see \emph{e.g}~\cite[Corollary I.2.23]{Rooduijn2024}). 
\begin{corollary}
\label{cor:regcompleteness}
If $\Gamma \Rightarrow_A \Delta$ has an $\mathsf{SGKAT}^\infty$-proof, then it has a regular one.
\end{corollary}
We define a \emph{cyclic $\mathsf{SGKAT}$-proof} as a regular $\mathsf{SGKAT}^\infty$-proof. Cyclic proofs can be equivalently described using finite trees with back edges, but this is not needed for the purposes of the present paper. 
\section{Completeness and complexity}
\label{sec:sgkatcompleteness}
In this section we prove the completeness of $\mathsf{SGKAT}^\infty$. Our argument uses a proof search procedure, which we will show to induce a $\mathsf{NLOGSPACE}$ decision procedure for the language inclusion problem of $\gkat$ expressions. The material in this section is again inspired by~\cite{DasPous17}, but requires several modifications to treat the tests present in $\gkat$. 

First note the following fact.
\begin{lemma}
	\label{lem:sgkatvalidconclusionsome}
	Any valid sequent is the conclusion of some rule application. 
\end{lemma}
Note that in the following lemma $A$ and $B$ may be distinct. 
\begin{lemma}
\label{lem:right-productivity}
Let $\pi$ be a derivation using only right logical rules and containing a branch of the form:
\begin{equation}
\tag{*}
\label{eqn:branch}
\AxiomC{$\Gamma \Rightarrow_B e^{(b)}, \Delta$}
\noLine
\UnaryInfC{$\vdots$}
\RightLabel{\footnotesize $(b)$-$r$}
\UnaryInfC{$\Gamma \Rightarrow_A e^{(b)}, \Delta$}
\DisplayProof
\end{equation}
such that (1) $\Gamma \Rightarrow_A e^{(b)}, \Delta$ is valid, and (2) every succedent on the branch has $e^{(b)}, \Delta$ as a final segment. Then $\Gamma \Rightarrow_B 0$ is valid.
\end{lemma}
\begin{proof}
    We claim that $e^{(b)} \Rightarrow_B 0$ is provable. We will show this by exploiting the symmetry of the left and right logical rules of $\mathsf{SGKAT}$ (cf. Remark~\ref{rmk:symmetry}). Since on the branch (\ref{eqn:branch}) every rule is a right logical rule, and $e^{(b)}, \Delta$ is preserved throughout, we can construct a derivation $\pi'$ of $e^{(b)} \Rightarrow_B 0$ from $\pi$ by applying the analogous left logical rules to $e^{(b)}$. Note that the set of atoms $B$ precisely determines the branch (\ref{eqn:branch}), in the sense that for every leaf $\Gamma \Rightarrow_C \Theta$ of $\pi$ it holds that $C \cap B = \emptyset$. Hence, as the root of $\pi'$ is $e^{(b)} \Rightarrow_B 0$, every branch of $\pi'$ except for the one corresponding to (\ref{eqn:branch}) can be closed directly by an application of $\bot$. The branch corresponding to (\ref{eqn:branch}) is of the form
\begin{equation*}
\tag{*}
\AxiomC{$e^{(b)} \Rightarrow_B 0$}
\noLine
\UnaryInfC{$\vdots$}
\RightLabel{\footnotesize $(b)$-$l$}
\UnaryInfC{$e^{(b)} \Rightarrow_B 0$}
\DisplayProof
\end{equation*}
and can thus be closed by a back edge. The resulting finite tree with back edges clearly represents an $\mathsf{SGKAT}^\infty$-proof. 

Now by soundness, we have $B \diamond \llbracket e^{(b)} \rrbracket = \emptyset$. Moreover, by the invertibility of the right logical rules and hypothesis (1), we get 
\[
B \diamond \llbracket \Gamma \rrbracket \subseteq B \diamond \llbracket e^{(b)}\rrbracket \diamond \llbracket \Delta \rrbracket = \emptyset,
\]
as required.
\end{proof}
\begin{lemma}
\label{lem:sgkatmonotone-r}
Let $(\Gamma_n \Rightarrow_{A_n} \Delta_n)_{n \in \omega}$ be an infinite branch of some $\mathsf{SGKAT}^\infty$-derivation on which the rule $(b)$-$r$ is applied infinitely often. Then there are $n, m$ with $n < m$ such that the following hold:
\begin{enumerate}[label = (\roman*), leftmargin=2.5em]
\item the sequents $\Gamma_n \Rightarrow_{A_n} \Delta_n$ and $\Gamma_m \Rightarrow_{A_m} \Delta_m$ are equal;
\item the sequent $\Gamma_n \Rightarrow_{A_n} \Delta_n$ is the conclusion of $(b)$-$r$ in $\pi$;
\item for every $i \in [n, m)$ it holds that $\Delta_n$ is a final segment of $\Delta_i$.
\end{enumerate}
\end{lemma}
\begin{proof}
First note that $\mathsf{k}_0$ is not applied on this branch, because if it were then there could not be infinitely many applications of $(b)$-$r$.

Since the proof is finite-state (cf. Corollary \ref{cor:anyfrugal}), there must be a $k \geq 0$ be such that every $\Delta_i$ with $i \geq k$ occurs infinitely often on the branch above. Denote by $|\Delta|$ the length of a given list $\Delta$ and let $l$ be minimum of $\{|\Delta_i| : i \geq k\}$. In other words, $l$ is the minimal length of the $\Delta_i$ with $i \geq k$. 

To prove the lemma, we first claim that there is an $n \geq k$ such that $|\Delta_n| = l$ and the leftmost expression in $\Delta_n$ is of the form $e^{(b)}$ for some $e$. Suppose, towards a contradiction, that this is not the case. Then there must be a $u \geq k$ such that $|\Delta_u| = l$ and the leftmost expression in $\Delta_u$ is \emph{not} of the form $e^{(b)}$ for any $e$. Note that $(b)$-$r$ is the only rule apart from $\mathsf{k}_0$ that can increase the length of the succedent (when read bottom-up). It follows that for no $w \geq u$ the leftmost expression in $\Delta_w$ is of the form $e^{(b)}$, contradicting the fact that $(b)$-$r$ is applied infinitely often.

Now let $n \geq k$ be such that $|\Delta_n| = l$ and the leftmost expression of $\Delta_n$ is $e^{(b)}$. Since the rule $(b)$-$r$ must at some point after $\Delta_n$ be applied to $e^{(b)}$, we may assume without loss of generality that $\Gamma_n \Rightarrow_{A_n} \Delta_n$ is the conclusion of an application of $(b)$-$r$. By the pigeonhole principle, there must be an $m > n$ such that $\Gamma_n \Rightarrow_{A_n} \Delta_n$ and $\Gamma_m \Rightarrow_{A_m} \Delta_m$ are the same sequents. We claim that these sequents satisfy the three properties above. Properties (i) and (ii) directly hold by construction. Property (iii) follows from the fact that $\Delta_n$ is of minimal length and has $e^{(b)}$ as leftmost expression. 
\end{proof}
With the above lemmas in place, we are ready for the completeness proof.
\begin{theorem}[Completeness]
\label{thm:nonregcompleteness}
Every valid sequent is provable in $\mathsf{SGKAT}^\infty$.
\end{theorem}
\begin{proof}
Given a valid sequent, we do a bottom-up proof search with the following strategy. Throughout the procedure all leaves remain valid, in most cases by an appeal to invertibility. 
\begin{enumerate}
	\item Apply left logical rules as long as possible. If this stage terminates, it will be at a leaf of the form $\Gamma \Rightarrow_A \Delta$, where $\Gamma$ is exposed. We then go to stage (2). If left logical rules remain applicable, we stay in this stage (1) forever and create an infinite branch. 
	
	\item Apply right logical rules until one of the following happens:
	\begin{enumerate}
		\item We reach a leaf at which no right logical rule can be applied. This means that the leaf must be a valid sequent of the form $\Gamma \Rightarrow_A \Delta$ such that $\Gamma$ is exposed, and $\Delta$ is either exposed or begins with a test $b$ such $A \restriction b \not= A$. We go to stage (4).
		\item If (a) does not happen, then at some point we must reach a valid sequent of the $\Gamma\Rightarrow_A e^{(b)}, \Delta$ which together with an ancestor satisfies properties (i) - (iii) of Lemma \ref{lem:sgkatmonotone-r}. In this case Lemma \ref{lem:right-productivity} is applicable. Hence we must be at a leaf of the form $\Gamma \Rightarrow_A  e^{(b)}, \Delta$ such that $e^{(b)} \Rightarrow_A 0$ is valid. We then go to stage (3). 
	\end{enumerate}
	Since at some point either (a) or (b) must be the case, stage (2) always terminates. 
	
	\item We are at a valid leaf of the form $\Gamma \Rightarrow_A e^{(b)}, \Delta$, where $\Gamma$ is exposed. If $A = \emptyset$, we apply $\bot$. Otherwise, if $A \not= \emptyset$, we use the validity of $\Gamma \Rightarrow_A e^{(b)}, \Delta$ and $e^{(b)} \Rightarrow_A 0$ to find:
	\[
	A \fusion \int{\Gamma} \subseteq A \fusion \int{e^{(b)}} \fusion \int{\Delta} = \emptyset.
	\] 
	We claim that $\int{\Gamma} = \emptyset$. Indeed, suppose towards a contradiction that $\alpha x \in \int{\Gamma}$. By the exposedness of $\Gamma$ and item (i) of Lemma \ref{lem:exposed}, we would have $\beta x \in \int{\Gamma}$ for some $\beta \in A$, contradicting the statement above. Therefore, the sequent $\Gamma \Rightarrow_\mathsf{At} 0$ is valid. We apply the rule $\mathsf{k}_0$ and loop back to stage (1). 
	
	Stage (3) only comprises a single step and thus always terminates. 
	
	\item Let $\Gamma \Rightarrow_A \Delta$ be the current leaf. By construction $\Gamma \Rightarrow_A \Delta$ is valid, $\Gamma$ is exposed, and $\Delta$ is either exposed or begins with a test $b$ such that $A \restriction b \not= A$. Note that only rules $\mathsf{id}$, $\bot$, $\mathsf{k}$, and $\mathsf{k_0}$ can be applicable. By Lemma \ref{lem:sgkatvalidconclusionsome}, at least one of them must be applicable. If $\mathsf{id}$ is applicable, apply $\mathsf{id}$. If $\bot$ is applicable, apply $\bot$. If $\mathsf{k}$ is applicable, apply $\mathsf{k}$ and loop back to stage (1). Note that this application of $\mathsf{k}$ will have priority and is therefore invertible. 

	Finally, suppose that only $\mathsf{k}_0$ is applicable. We claim that, by validity, the list $\Gamma$ is not $\epsilon$. Indeed, since $A$ is non-empty, and $\Delta$ either begins with a primitive program $p$ or a test $b$ such that $A \restriction b \not= A$, the sequent
	\[
	\epsilon \Rightarrow_A \Delta
	\]
	must be invalid. Hence $\Gamma$ must be of the form $p, \Theta$. We apply $\mathsf{k}_0$, which has priority and thus is invertible, and loop back to stage (1). 
	
	Similarly to stage (3), stage (4) only comprises a single step and thus always terminates.
\end{enumerate}

We claim that the constructed derivation is fair for $(b)$-$l$. Indeed, every stage except stage (1) terminates. Therefore, every infinite branch must either eventually remain in stage (1), or pass through stages (3) or (4) infinitely often. Since $\mathsf{k}$ and $\mathsf{k}_0$ shorten the antecedent, and no left logical rule other than $(b)$-$l$ lengthens it, such branches must be fair.
\end{proof} 
By Corollary \ref{cor:regcompleteness} we obtain that the subset of cyclic $\mathsf{SGKAT}$-proofs is also complete. 
\begin{corollary}
Every valid sequent has a regular $\mathsf{SGKAT}^\infty$-proof. 
\end{corollary}
\begin{proposition}
The proof search procedure of Theorem \ref{thm:nonregcompleteness} runs in $\mathsf{coNLOGSPACE}$. Hence proof search, and thus also the language inclusion problem for $\gkat$-expressions, is in $\mathsf{NLOGSPACE}$.
\end{proposition}
\begin{proof}[sketch]
Assume without loss of generality that the initial sequent is of the form $e \Rightarrow_A f$. We non-deterministically search for a failing branch, at each iteration storing only the last sequent. By Lemma \ref{lem:realisationoftailgenerated} this can be done by storing two pointers to, respectively, the syntax trees $T_e$ and $T_f$, together with a set of atoms. The loop check of stage (2) can be replaced by a counter. Indeed, stage (2) must always hit a repetition after $|\mathsf{At}| \cdot |T_f|$ steps, where $m$ is the number of nodes in the syntax tree. After this repetition there must be a continuation that reaches a repetition to which Lemma \ref{lem:right-productivity} applies before this stage has taken $2 \cdot |\mathsf{At}| \cdot |T_f|$ steps in total. Finally, a global counter can be used to limit the depth of the search. Indeed, a failing branch needs at most one repetition (in stage (2), to which $\mathsf{k}_0$ is applied) and all other repetitions can be cut out. Hence if there is a failing branch, there must be one of size at most $4 \cdot |T_e| \cdot |\mathsf{At} | \cdot |T_f|$.
\end{proof}
\section{Conclusion and future work}
In this paper we have presented a non-well-founded proof system $\mathsf{SGKAT}^\infty$ for $\gkat$. We have shown that the system is sound and complete with respect to the language model. In fact, the fragment of \emph{regular} proofs is already complete, which means one can view $\mathsf{SGKAT}$ as a cyclic proof system. Our system is similar to the system for Kleene Algebra in~\cite{DasPous17}, but the deterministic nature of $\gkat$ allows us to use ordinary sequents rather than hypersequents. To deal with the tests of $\gkat$ every sequent is annotated by a set of atoms. Like in~\cite{DasPous17}, our completeness argument makes use of a proof search procedure. Here again the relative simplicity of $\gkat$ pays off: the proof search procedure induces an $\mathsf{NLOGSPACE}$ decision procedure, whereas that of Kleene Algebra is in $\mathsf{PSPACE}$. 

The most natural question for future work is whether our system could be used to prove the completeness of some (ordered)-algebraic axiomatisation of $\gkat$. We envision using the original $\gkat$ axioms (see~\cite[Figure 1]{Smolka20}), but basing it on \emph{inequational} logic rather than equational logic. This would allow one to use a \emph{least} fixed point rule of the form 
\[
\AxiomC{$eg +_b f \leq g$}
\UnaryInfC{$e^{(b)} f \leq g$}
\DisplayProof
\]
eliminating the need for a Salomaa-style side condition. We hope to be able to prove the completeness of such an inequational system by translating cyclic $\mathsf{SGKAT}$-proofs into well-founded proofs in the inequational system. This is inspired by the paper~\cite{Das18}, where a similar strategy is used to give an alternative proof of the left-handed completeness of Kleene Algebra.

Another relevant question is the exact complexity of the language inclusion problem for $\gkat$-expressions. We have obtained an upper bound of $\mathsf{NLOGSPACE}$, but do not know whether it is optimal.

Finally, it would be interesting to verify the conjecture in Remark \ref{rmk:earlyterm} above. \\

\begin{credits}
\ackname \ Jan Rooduijn thanks Anupam Das, Tobias Kapp\'e, Johannes Marti and Yde Venema for insightful discussions on the topic of this paper. Alexandra Silva wants to acknowledge Sonia Marin, who some years ago proposed a similar master project at UCL. We moreover thank the reviewers for their helpful comments, in particular for pointing out that our complexity result could be sharpened. Lastly, Jan Rooduijn is grateful for the inspiring four-week research visit at the Computer Science department of Cornell in the summer of 2022.
\end{credits}
\bibliographystyle{splncs04}
\bibliography{mybibliography}
\clearpage
\appendix
\setcounter{theorem}{0}
\setcounter{proposition}{0}
\setcounter{lemma}{0}
\setcounter{corollary}{0}
\section{Full proofs}
This appendix contains full versions of all the proofs that were either omitted or sketched in the body of the paper.
\setcounter{subsection}{1}
\subsection{... of Section \ref{sec:prelim}}
	\begin{lemma}
    \label{lem:fusion_n_symm}
    For any two languages $L, K$ of guarded strings, and primitive program $p$, we have:
    \begin{enumerate}[label = (\roman*)]
        \item $L^{n + 1} = L \fusion L^n$; \ \ \ \ (ii) $\int{p} \fusion L = \int{p} \fusion K$ implies $L = K$. 
    \end{enumerate}
	\end{lemma}
    \begin{proof}
	(i). Since $\mathsf{At}$ is the identity element for the fusion operator, we have
	\[
	L^{n + 1} = \mathsf{At} \fusion \underbrace{L \fusion \cdots \fusion L}_{\text{$n + 1$ times}} = L \fusion \mathsf{At} \fusion \underbrace{L \fusion \cdots \fusion L}_{\text{$n$ times}} = L \fusion L^n,
	\]
	as required. 

    (ii). Since $\int{p} = \{\alpha p \beta : \alpha, \beta \in \mathsf{At}\}$, we have
		\[
		\gamma y \in L \Leftrightarrow \gamma p \gamma y \in \int{p} \fusion L \Leftrightarrow \gamma p \gamma y \in \int{p} \fusion K \Leftrightarrow \gamma y \in K,
		\]
		as required. 
	\end{proof}
	\subsection{... of Section \ref{sec:thesgkatsystem}}
    \setcounter{lemma}{2}
	\begin{lemma}
Let $\Gamma$ and $\Delta$ be exposed lists of expressions. Then:
\begin{enumerate}[label = (\roman*)]
\item $\alpha x \in \int{\Gamma} \Leftrightarrow \beta x \in \int{\Gamma}$ for all $\alpha, \beta \in \mathsf{At}$
\item $\Gamma \Rightarrow_{\mathsf{At}} \Delta$ is valid if and only if $\Gamma \Rightarrow_A \Delta$ is valid for some $A \not= \emptyset$. 
\end{enumerate}
\end{lemma}
\begin{proof}
For item (i), we make a case distinction on whether $\Gamma = \epsilon$ or $\Gamma = p, \Theta$ for some list $\Theta$. If $\Gamma = \epsilon$, the result follows immediately from the fact that $\int{\epsilon} = \mathsf{At}$. If $\Gamma = p, \Theta$, we have $\int{\Gamma} = \int{p} \fusion \int{\Theta} = \{\gamma p \delta y : \gamma \in \mathsf{At}, \delta y \in \int{\Theta}\}$, which also suffices. 

For item (ii), the only non-trivial implication is the one from right to left. So suppose $\Gamma \Rightarrow_A \Delta$ for some $A \not= \emptyset$. Let $\alpha \in \mathsf{At}$ and let $\beta \in A$ be arbitrary. We find the required:
\begin{align*}
\alpha x \in \int{\Gamma} &\Rightarrow \beta x \in \int{\Gamma} &\text{(item (i))} \\
						  &\Rightarrow \beta x \in \int{\Delta} &\text{($\beta \in A$, hypothesis)} \\
						  &\Rightarrow \alpha x \in \int{\Delta}. &\text{(item (i))}
\end{align*}
\end{proof}
\subsection{... of Section \ref{sec:sgkatsoundess}}
\begin{lemma}
Let $A$ be a set of atoms, let $b$ be a test, and let $\Theta$ be a list of expressions. We have:
\begin{enumerate}[label = (\roman*)]
	\item $\int{e +_b f, \Theta} = (\int{b} \fusion \int{e, \Theta}) \cup (\int{\overline b} \fusion \int{f, \Theta})$;
	\item $\int{e^{(b)}, \Theta} = (\int{b} \fusion \int{e, e^{(b)}, \Theta}) \cup (\int{\overline{b}} \fusion \int{\Theta})$. 
\end{enumerate}
\end{lemma}
\begin{proof}
Both items are shown by simply unfolding the definitions. We will use the fact $\fusion$ distributes over $\cup$. Note that $\cup$ is not the same as \emph{guarded} union, over which $\fusion$ is merely right-distributive. 

For the first item, we calculate
\begin{align*}
\int{e +_b f, \Theta}  &=  \int{e +_b f} \fusion \int{\Theta} &\text{(sequent interpretation)} \\
&= ((\int{b} \fusion \int{e}) \cup (\overline{\int{b}} \fusion \int{f}))\fusion \int{\Theta} &\text{(interpretation of $+_b$)} \\
&= (\int{b} \fusion \int{e} \fusion \int{\Theta}) \cup (\overline{\int{b}} \fusion \int{f} \fusion \int{\Theta})  &\text{($\fusion$ distributes over $\cup$)} \\
&= (\int{b} \fusion \int{e} \fusion \int{\Theta}) \cup (\int{\overline b} \fusion \int{f} \fusion \int{\Theta}) &\text{($\overline{\int{b}} = \int{\overline b}$)} \\
&= (\int{b} \fusion \int{e, \Theta}) \cup (\int{\overline b} \fusion \int{f,\Theta}). &\text{(sequent interpretation)} 
\end{align*}
For the second item, we have
\begin{align*}
\int{e^{(b)}, \Theta}  &= \int{e^{(b)}} \fusion \int{\Theta} &\text{(sequent int.)} \\
&= \bigcup_{n \geq 0} (\int{b} \fusion \int{e})^n \fusion \overline{\int{b}} \fusion \int{\Theta} &\text{(int. $-^{(b)}$)} \\
&= (\bigcup_{n \geq 1} (\int{b} \fusion \int{e})^n \cup \mathsf{At}) \fusion \overline{\int{b}} \fusion \int{\Theta} &\text{(split ${\bigcup}$)} \\
&= (\bigcup_{n \geq 1} (\int{b} \fusion \int{e})^n \fusion \overline{\int{b}} \fusion \int{\Theta}) \cup (\mathsf{At} \fusion \overline{\int{b}} \fusion \int{\Theta})  &\text{($\fusion$ dist. $\cup$)} \\
&= (\bigcup_{n \geq 1} (\int{b} \fusion \int{e})^n \fusion \int{\overline b} \fusion \int{\Theta}) \cup (\int{\overline b} \fusion \int{\Theta}) &\text{($\overline{\int{b}} = \int{\overline b}$)} \\
&= (\bigcup_{n \geq 0} \int{b} \fusion \int{e} \fusion (\int{b} \fusion \int{e})^{n} \fusion \int{\overline b} \fusion \int{\Theta}) \cup (\int{\overline b} \fusion \int{\Theta}) &\text{(Lem. \ref{lem:fusion_n_symm}.(i))} \\
&= (\int{b} \fusion \int{e} \fusion \bigcup_{n \geq 0} (\int{b} \fusion \int{e})^{n} \fusion \int{\overline b} \fusion \int{\Theta}) \cup (\int{\overline b} \fusion \int{\Theta}) &\text{($\fusion$ dist. $\bigcup$)} \\
&= (\int{b} \fusion \int{e} \fusion \int{e^{(b)}} \fusion \int{\Theta}) \cup (\int{\overline b} \fusion \int{\Theta}) &\text{(int. $-^{(b)}$)} \\
&= (\int{b} \fusion \int{e, e^{(b)}, \Theta}) \cup (\int{\overline b} \fusion \int{\Theta}), &\text{(sequent int.)}
\end{align*}
as required.
\end{proof}
\setcounter{lemma}{4}
\begin{proposition}
	Every rule of $\mathsf{SGKAT}$ is sound. Moreover, every rule is invertible except for $\mathsf{k}$ and $\mathsf{k}_0$, which are invertible whenever they have priority. 
\end{proposition}
\begin{proof}
We will cover the rules of $\mathsf{SGKAT}$ one-by-one. 
\begin{itemize}[labelwidth=4em,leftmargin =\dimexpr\labelwidth+\labelsep\relax - 1em]
\item[($b$-$l$)] This is immediate by Lemma \ref{lem:soundnessaux}.1.

\item[($b$-$r$)] We have:
\begin{align*}
A \fusion \int{\Gamma} \subseteq  \int{b, \Delta} &\Leftrightarrow A \fusion\int{\Gamma} \subseteq  \int{b} \fusion \int{\Delta} &\text{(sequent int.)} \\
	&\Leftrightarrow A \restriction b \fusion \int{\Gamma} \subseteq \int{b} \fusion \int{\Delta} &\text{(by ($\dagger$))} \\ 
	&\Leftrightarrow A \restriction b \fusion \int{\Gamma} \subseteq \int{\Delta} &\text{($A \restriction b \subseteq \int{b}$)} \\
	&\Leftrightarrow A \fusion \int{\Gamma} \subseteq \int{\Delta} &\text{(by ($\dagger$))}
\end{align*}
\item[($\cdot$-$l$)] Immediate, since $A \fusion \int{e \cdot f, \Gamma} = A \fusion \int{e, f, \Gamma}$.
\item[($\cdot$-$r$)] Likewise, but by $\int{e \cdot f, \Delta} = \int{e, f, \Delta}$. 
\item[($+_b$-$l$)] This follows directly from the fact that
\begin{align*}
A \fusion \int{e +_b f, \Gamma}  &= A \fusion \int{e +_b f} \fusion \int{\Gamma} &\text{(sequent int.)} \\
&=  A \fusion ((\int{b} \fusion \int{e, \Gamma}) \cup (\int{\overline b} \fusion \int{f, \Gamma})) &\text{(Lem. \ref{lem:soundnessaux}.(ii))} \\
&=  (A \fusion \int{b} \fusion \int{e, \Gamma}) \cup (A \fusion \int{\overline b} \fusion \int{f, \Gamma}) &\text{(distrib.)} \\
&= (A \restriction b \fusion \int{e, \Gamma}) \cup (A \restriction \overline b \fusion \int{f, \Gamma}) &\text{(Lem. \ref{lem:soundnessaux}.(i))} 
\end{align*}
\item[($+_b$-$r$)] We find 
\begin{align*}
	&A \fusion \int{\Gamma} \subseteq \int{e +_b f} \fusion \int{\Delta} \\\
	&\ \  \Leftrightarrow A \fusion \int{\Gamma} \subseteq (\int{b} \fusion \int{e, \Delta}) \cup (\int{\overline b} \fusion \int{f, \Delta}) \\
	&\ \ \ \ \Leftrightarrow A \restriction b \fusion \int{\Gamma} \subseteq \int{e, \Delta} \text{ or } A \restriction \overline b \subseteq \int{f, \Delta},
	\end{align*}
	where the first equivalence holds due to Lemma \ref{lem:soundnessaux}.(ii), and the second due to $A \fusion \int{\Gamma} = (\int{b} \fusion A \fusion \int{\Gamma}) \cup (\int{\overline b} \fusion A \fusion \int{\Gamma})$ and Lemma \ref{lem:soundnessaux}.(i).
\item[($(b)$-$l$)] This follows directly from the fact that
\begin{align*}
A \fusion \int{e^{(b)}, \Gamma} &= A \fusion \int{e^{(b)}} \fusion \int{\Gamma} &\text{(sequent int.)} \\
&=  A \fusion ((\int{b} \fusion \int{e, e^{(b)}, \Gamma}) \cup (\int{\overline b} \fusion \int{f, \Gamma})) &\text{(Lem. \ref{lem:soundnessaux}.(iii))} \\
&=  (A \fusion \int{b} \fusion \int{e, e^{(b)}, \Gamma}) \cup (A \fusion \int{\overline b} \fusion \int{f, \Gamma}) &\text{(distrib.)} \\
&= (A \restriction b \fusion \int{e, e^{(b)}, \Gamma}) \cup (A \restriction \overline b \fusion \int{f, \Gamma}) &\text{(Lem. \ref{lem:soundnessaux}.(i))} 
\end{align*}
\item[($(b)$-$r$)] We find 
\begin{align*}
	&A \fusion \int{\Gamma} \subseteq \int{e^{(b)}, \Delta} \\
	&\ \ \Leftrightarrow A \fusion \int{\Gamma} \subseteq (\int{b} \fusion \int{e, e^{(b)}, \Delta}) \cup (\int{\overline b} \fusion \int{\Delta}) \\
	&\ \ \ \ \Leftrightarrow  A \restriction b \fusion \int{\Gamma} \subseteq \int{b} \fusion \int{e, e^{(b)}, \Delta} \text{ and } A \restriction \overline b \subseteq \int{\overline b} \fusion \int{\Delta},
	\end{align*}
	where the first equivalence holds due to Lemma \ref{lem:soundnessaux}.3, and the second due to $A \fusion \int{\Gamma} = (\int{b} \fusion A \fusion \int{\Gamma}) \cup (\int{\overline b} \fusion A \fusion \int{\Gamma})$ and Lemma \ref{lem:soundnessaux}.1.
\item[($\mathsf{id}$)] This follows from $A \fusion \int{1} = A \fusion \mathsf{At} = A \subseteq \mathsf{At} = \int{1}$.  
\item[($\bot$)] We have $\emptyset \fusion \int{\Gamma} = \emptyset \subseteq \int{\Delta}$. 
	\item[($\mathsf{k}$)] Suppose first that some application of $\mathsf{k}$ does \emph{not} have priority. The only rule of higher priority than $\mathsf{k}$ which can have a conclusion of the form $p, \Gamma \Rightarrow_A p, \Delta$ is $\bot$, whence we must have $A = \emptyset$. As shown in the previous case, this conclusion must be valid. Hence under this restriction the rule application is vacuously sound. It is, however, not invertible, as the following rule instance demonstrates
	\[
	\AxiomC{$1 \Rightarrow_\mathsf{At} 0$}
	\LeftLabel{$\mathsf{k}$}
	\UnaryInfC{$p, 1 \Rightarrow_\emptyset p, 0$}
	\DisplayProof
	\]

	Next, suppose that some application of $\mathsf{k}$ does have priority. This means that the set $A$ of atoms in the conclusion $p, \Gamma \Rightarrow_A p, \Delta$ is \emph{not} empty. We will show that under this restriction the rule is both sound and invertible. Let $\alpha \in A$. We have
	\begin{align*}
	A \fusion \int{p, \Gamma} \subseteq \int{p, \Delta} &\Leftrightarrow A \fusion \int{p} \fusion \int{\Gamma} \subseteq \int{p} \fusion \int{\Delta} &\text{(seq. int.)}\\
	&\Leftrightarrow {\alpha} \fusion \int{p} \fusion \int{\Gamma} \subseteq \int{p} \fusion \int{\Delta} &\text{($\alpha \in A$, Lem. \ref{lem:exposed})}\\
	&\Leftrightarrow \int{p} \fusion \int{\Gamma} \subseteq \int{p} \fusion \int{\Delta} &\text{(Lem. \ref{lem:exposed})} \\
	&\Leftrightarrow \int{\Gamma} \subseteq \int{\Delta}, &\text{(Lem. \ref{lem:fusion_n_symm}.(ii))}
	\end{align*}
	as required. 
	\item[($\mathsf{k}_0$)] For the final rule $\mathsf{k}_0$, we will first show the soundness of all instances, and then the invertibility of those instances which have priority. For soundness, suppose that the premiss is valid. Since
	\[
	\int{\Gamma} = \mathsf{At} \fusion \int{\Gamma} \subseteq \int{0} = \emptyset,
	\]
	it follows that $\int{\Gamma} = \emptyset$. Hence
	\[
	A \fusion \int{p, \Gamma} = A \fusion \int{p} \fusion \int{\Gamma} =  A \fusion \int{p} \fusion \emptyset = \emptyset \subseteq \int{\Delta},
	\]
	as required. 

	For invertibility, suppose that some instance of $\mathsf{k}_0$ has priority. Then the conclusion $p, \Gamma \Rightarrow_A \Delta$ cannot be the conclusion of any other rule application. 

	Suppose that $p, \Gamma \Rightarrow_A \Delta$ is valid. We wish to show that $\Gamma \Rightarrow_\mathsf{At} 0$ is valid, or, in other words, that $\int{\Gamma} = \emptyset$. 

	First note that, as in the previous case, from the assumption that our instance of $\mathsf{k}_0$ has priority, it follows that $A \not= \emptyset$. 

	We now make a case distinction on the shape of $\Delta$. Suppose first that $\Delta = \epsilon$. Then 
	\[
	 A \fusion \int{p, \Gamma} \subseteq \int{\Delta} = \int{\epsilon} = \mathsf{At}. 
	\]
	As $A \fusion \int{p, \Gamma} = \{\alpha p \beta x : \alpha \in A \text{ and } \beta x \in \int{\Gamma}\}$, we must have $\int{\Gamma} = \emptyset$.

	Next, suppose that $\Delta$ has a leftmost expression $e$. By the assumption that the rule instance has priority, we know that $e$ is not of the form $e_0 \cdot e_1$, $e_0 +_b e_1$, or $e^{(b)}$, for otherwise a right logical rule could be applied. Hence, the expression $e$ must either be a test or a primitive program. 

	If $e$ is a test, say $b$, we know that $A \restriction b \not= A$, for otherwise $b$-$r$ could be applied. Recall that it suffices to show that $\int{\Gamma} = \emptyset$. So suppose, towards a contradiction, that there is some $\beta x \in \int{\Gamma}$. Let $\alpha \in A$ such that $\alpha \not \leq b$. Then $\alpha p \beta x \in \int{p, \Gamma} \subseteq \int{\Delta}$. But this contradicts the fact that $\int{\Delta} \subseteq \{\alpha y : \alpha \leq b\}$.

	Finally, suppose that $e$ is a primitive program, say $q$. Write $\Delta = q, \Theta$. First note that assumption that the rule instance has priority implies $p \not= q$, for otherwise the rule $\mathsf{k}$ could be applied. We have:
	\[
	A \fusion \int{p, \Gamma} \subseteq \int{\Delta} = \{\alpha q \beta x : \beta x \in \int{\Theta}\},
	\]
	As $A \fusion \int{p, \Gamma} = \{\alpha p \beta x : \alpha \in A \text{ and } \beta x \in \int{\Gamma}\}$ and $p \not= q$, we again find that $\int{\Gamma} = \emptyset$. 
\end{itemize}
This finishes the proof. 
\end{proof}
\begin{lemma}
	For every $n \in \mathbb{N}$: if we have $\mathsf{SGKAT} \vdash^\infty e^{(b)}, \Gamma \Rightarrow_A \Delta$, then we also have $\mathsf{SGKAT} \vdash^\infty [e^{(b)}]^{n}, \Gamma \Rightarrow_A \Delta$.
\end{lemma}
\begin{proof}
	We assume that $A \not= \emptyset$, for otherwise the lemma is trivial. Let $\pi$ be the assumed $\mathsf{SGKAT}^\infty$-proof of $e^{(b)}, \Gamma \Rightarrow_A \Delta$. Note that, since all succedents referred to in the lemma are equal to $\Delta$, it suffices to prove the lemma under the assumption that the last rule applied in $\pi$ is \emph{not} a right logical rule. Hence, we may assume that the last rule applied in $\pi$ is $(b)$-$l$, for that is the only remaining rule with a sequent of this shape as conclusion. This means that $\pi$ is of the form:
	\[
	\AxiomC{$\pi_1$}
	\noLine
	\UnaryInfC{$e , e^{(b)},\Gamma \Rightarrow_{A \restriction b} \Delta$}
	\AxiomC{$\pi_2$}
	\noLine
	\UnaryInfC{$\Gamma \Rightarrow_{A \restriction \overline b} \Delta$}
	\RightLabel{\footnotesize $(b)$-$l$}
	\BinaryInfC{$e^{(b)}, \Gamma \Rightarrow_A \Delta$}
	\DisplayProof
	\]
	We show the lemma by induction on $n$. For the induction base, we take the following proof:
	\[
	\AxiomC{$\pi_2$}
	\noLine
	\UnaryInfC{$\Gamma \Rightarrow_{A \restriction \overline b} \Delta$}
	\RightLabel{\footnotesize $\overline b$-$l$}
	\UnaryInfC{$[e^{(b)}]^{0}, \Gamma \Rightarrow_A \Delta$}
	\DisplayProof
	\]
	For the inductive step $n + 1$, we construct from $\pi_1$ a proof $\tau$ of $e , [e^{(b)}]^{n},\Gamma \Rightarrow_{A \restriction b} \Delta$. To that end, we first replace in $\pi_1$ every occurrence of $e^{(b)}, \Gamma$ as a final segment of the antecedent by $e^{(b)^{n}}, \Gamma$ and cut off all branches at sequents of the form $[e^{(b)}]^{n}, \Gamma \Rightarrow_B \Theta$. This may be depicted as follows, where to the left of the arrow $\leadsto$ we have a branch of $\pi_1$, and to right the resulting branch of $\tau$. 
	\begin{align*}
	\AxiomC{$\vdots$}
	\UnaryInfC{$e^{(b)}, \Gamma \Rightarrow_{B} \Theta$}
	\noLine
	\UnaryInfC{$\vdots$}
	\UnaryInfC{$e , e^{(b)},\Gamma \Rightarrow_{A \restriction b} \Delta$}
	\DisplayProof
	\leadsto 
	\AxiomC{$\phantom{\vdots}$}
	\UnaryInfC{$[e^{(b)}]^{n}, \Gamma \Rightarrow_{B} \Theta$}
	\noLine
	\UnaryInfC{$\vdots$}
	\UnaryInfC{$e , [e^{(b)}]^{n},\Gamma \Rightarrow_{A \restriction b} \Delta$}
	\DisplayProof
	\end{align*} 
	Note that every remaining infinite branch in the resulting derivation $\tau$ satisfies the fairness condition. Therefore, to turn $\tau$ into a proper $\mathsf{SGKAT}^\infty$-proof, we only need to close each open leaf, which by construction is of the form $[e^{(b)}]^{n}, \Gamma \Rightarrow_B \Delta$. Note that $\pi_1$ must contain a proof of $e^{(b)}, \Gamma \Rightarrow_B \Delta$, whence by the induction hypothesis the sequent $[e^{(b)}]^{n}, \Gamma \Rightarrow_B \Delta$ is provable. We can thus close the leaf by simply appending the witnessing proof. 

	Letting $\tau$ be the resulting proof, we finish the induction step by taking:
	\[
	\AxiomC{$\tau$}
	\noLine
	\UnaryInfC{$e, [e^{(b)}]^{n}, \Gamma \Rightarrow_{A \restriction b} \Delta$}
	\RightLabel{\footnotesize $b$-$l$}
	\UnaryInfC{$[e^{(b)}]^{n+1}, \Gamma \Rightarrow_A \Delta$}
	\DisplayProof
	\]
	which gives us the required $\mathsf{SGKAT}^\infty$-proof. 
\end{proof}
\begin{lemma}
	$\wwh([e^{(b)}]^{n}, \Gamma) < \wwh(e^{(b)}, \Gamma)$ for every $n \in \mathbb{N}$. 
\end{lemma}
\begin{proof}
	Let $k := \wh(e^{(b)})$. Note that the maximum while-height in $[e^{(b)}]^{n}$ is that of $e$. Hence, we have $\wwh([e^{(b)}]^{n})(k) = 0 < 1 = \wwh{(e^{(b)})}(k)$. Therefore:
	\begin{align*}
	\wwh{([e^{(b)}]^{n}, \Gamma)}(k) &= \wwh([e^{(b)}]^{n})(k) + \wwh(\Gamma)(k) \\
	&<  \wwh(e^{(b)})(k) + \wwh(\Gamma)(k) = \wwh{(e^{(b)}, \Gamma)}(k).
	\end{align*}
	Hence $\wwh{([e^{(b)}]^{n}, \Gamma)} \not = \wwh(e^{(b)}, \Gamma)$. Now suppose that for some $l \in \mathbb{N}$ we have $\wwh([e^{(b)}]^{n}, \Gamma)(l) > \wwh(e^{(b)}, \Gamma)(l)$. We leave it to the reader to verify that in this case we must have $l < k$. As $\wwh{([e^{(b)}]^{n}, \Gamma)}(k) < \wwh(e^{(b)}, \Gamma)(k)$, we find $\wwh{([e^{(b)}]^{n}, \Gamma)} < \wwh(e^{(b)}, \Gamma)$.
\end{proof}
\setcounter{subsection}{5}
\subsection{... of Section \ref{sec:sgkatcompleteness}}
\setcounter{lemma}{7}
\begin{lemma}
	Any valid sequent is the conclusion of some rule application. 
\end{lemma}
\begin{proof}
	We prove this lemma by contraposition. So suppose $\Gamma \Rightarrow_A \Delta$ is \emph{not} the conclusion of \emph{any} rule application. We make a few observations:
	\begin{itemize}
		\item Both $\Gamma$ and $\Delta$ are exposed, for otherwise $\Gamma \Rightarrow_A \Delta$ would be the conclusion of an application of a left, respectively right, logical rule.
		\item $A$ is non-empty, for otherwise $\Gamma \Rightarrow_A \Delta$ would be the conclusion of an application of $\bot$.
		\item The leftmost expression of $\Gamma$ is not a primitive program, for otherwise our sequent $\Gamma \Rightarrow_A \Delta$ would be the conclusion of an application of $\mathsf{k_0}$.
		\item The leftmost expression of $\Delta$ \emph{is} a primitive program, for otherwise, by the previous items, the sequent $\Gamma \Rightarrow_A \Delta$ would be the conclusion of an application of $\mathsf{id}$. 
	\end{itemize}
	Hence $\Gamma \Rightarrow_A \Delta$ is of the form $\epsilon \Rightarrow_A p, \Theta$. Let $\alpha \in A$. Then $\alpha \in A \fusion \int{\epsilon}$. However, since $\alpha$ is not of the form $\beta p \gamma y$, we have $\alpha \notin \int{p, \Theta}$. This shows that $\Gamma \Rightarrow_A \Delta$ is not valid, as required. 
\end{proof}
\begin{proposition}
The proof search procedure of Theorem \ref{thm:nonregcompleteness} runs in $\mathsf{coNLOGSPACE}$. Hence proof search, and thus also the language inclusion problem for $\gkat$-expressions, is in $\mathsf{NLOGSPACE}$.
\end{proposition}
\begin{proof}
The exact $\mathsf{coNLOGSPACE}$ procedure goes as follows. Without loss of generality we assume that the input is a sequent of the form $e \Rightarrow_A f$, given as a triple consisting of the syntax tree $T_e$ of $e$, the set of atoms $A$, and the syntax tree of $T_f$ of $f$. During the procedure we store:
\begin{itemize}
    \item One pointer to $T_e$, initialised at the root. 
    \item One pointer to $T_f$, initialised at the root.
    \item A set of atoms, initially $A$.
\end{itemize}
Note that by Lemma \ref{lem:realisationoftailgenerated} these data suffice to completely describe a sequent. Moreover, we store:
\begin{itemize}
    \item A counter keeping track of how many steps are taken, initially $1$. 
    \item A variable $s \in \{1, 2, 3, 4\}$ indicating at which of the four stages of the procedure we currently are.
\end{itemize}
The idea is to non-deterministically apply the proof search procedure of Theorem \ref{thm:nonregcompleteness}. Note, however, that the fact that we are only storing a single sequent prevents us from performing the loop check that happens in stage (2). As a remedy, we additionally store
\begin{itemize}
    \item A second counter, keeping track of how many steps are taken in stage (2) (always resetting at the beginning of stage (2)).
\end{itemize}
At each iteration we apply the rule dictated by the proof search procedure, non-deterministically choosing a premiss and updating the data accordingly. Our non-deterministic procedure searches for a failing branch (\emph{i.e.} one to which no rule can be applied). Whenever the first counter reaches $4 \cdot |T_e| \cdot |\mathsf{At}| \cdot |T_f|$, it gives up the search. Moreover, whenever the second counter reaches $2 \cdot |\mathsf{At}| \cdot |T_f|$, it applies $\mathsf{k}_0$ (this replaces the loop check of stage (2)). 

We claim that this procedure succeeds in finding a failing branch if it exists. We first justify the application of $\mathsf{k}_0$ when the second counter hits its limit. Consider a list of sequents $\Gamma \Rightarrow_{A_0} \Delta_0, \ldots, \Gamma \Rightarrow_{A_n} \Delta_n$ in stage (2), where $n = |\mathsf{At}| \cdot |T_f|$. By the pidgeonhole principle, this list must repeat some sequent. Moreover, because the only rule that grows succedents is $(b)$-$r$, the segment between the two repetitions must contain a succedent of the form $e^{(b)}, \Delta$ of minimal length. But then there is a variant of this branch, of length at most $2n$, where the succedent $e^{(b)}, \Delta$ is repeated and satisfies the conditions of Lemma \ref{lem:right-productivity}.

Finally, we argue that if there is a failing branch, then there is one of length smaller than the limit of the first counter. This follows from the fact that the second counter can only hit its limit once (because afterwards the succedent becomes trivial) and every repetition that happens before or after the corresponding stage (2) can be cut out of the branch. 
\end{proof}
\end{document}